\documentclass[sn-mathphys-num]{sn-jnl}

\usepackage{graphicx}
\usepackage{amsmath,amssymb,amsfonts}
\usepackage{mathtools}
\usepackage{bm}
\usepackage{booktabs}
\usepackage{array}
\usepackage{microtype}
\usepackage{tikz}
\usetikzlibrary{arrows.meta,calc,positioning,fit}
\definecolor{geomblue}{RGB}{37,99,164}
\definecolor{geomorange}{RGB}{214,113,35}
\definecolor{geomred}{RGB}{174,43,50}
\definecolor{geomgray}{RGB}{84,92,101}
\definecolor{resolvedgreen}{RGB}{0,128,0}

\theoremstyle{thmstyleone}
\newtheorem{proposition}{Proposition}

\theoremstyle{thmstyletwo}

\theoremstyle{thmstylethree}

\newcommand{\dd}{\mathrm d}
\newcommand{\trans}{\mathsf T}
\newcommand{\Vspace}{\mathsf V}
\newcommand{\Hspace}{\mathsf H}
\newcommand{\Hor}{\operatorname{Hor}}

\newcommand{\Cov}{\operatorname{Cov}}
\newcommand{\Var}{\operatorname{Var}}
\newcommand{\argmin}{\operatorname*{arg\,min}}
\newcommand{\id}{\operatorname{id}}
\newcommand{\diag}{\operatorname{diag}}

\newcommand{\pr}{\operatorname{pr}}
\newcommand{\GO}{\mathrm{GO}}
\newcommand{\NP}{\mathrm{NP}}
\newcommand{\prof}{\mathrm{prof}}
\newcommand{\aux}{\mathrm{aux}}

\newcommand{\stat}{\mathrm{stat}}
\newcommand{\syst}{\mathrm{syst}}
\newcommand{\E}{\mathbb E}

\newcommand{\given}{\,;\,}

\makeatletter
\let\orcidlogo\@undefined
\makeatother
\usepackage{orcidlink}

\begin{document}

\title[The geometry of uncertainty decomposition in profile-likelihood fits]
{The geometry of uncertainty decomposition in profile-likelihood fits}

\author{\fnm{Rafael} \sur{Coelho Lopes de S\'a}\, \orcidlink{0000-0001-5200-9195}}
\email{rclsa@umass.edu}

\affil{\orgdiv{Department of Physics},
  \orgname{University of Massachusetts},
  \orgaddress{\city{Amherst}, \postcode{01003}, \state{MA},
  \country{USA}}}

\abstract{Uncertainty decompositions in profile-likelihood fits are commonly reported through nuisance-parameter impacts, although shifting a fitted parameter and fluctuating the observation that constrains it answer different questions. Recently, Pinto et al. provided an explicit construction for uncertainty decomposition based on fluctuating the observations and argued that such a construction allows for a cleaner interpretation of systematic uncertainties in terms of physical sources. We provide a geometric description of the distinction between the two methods by using a coupled pair of fiber bundles. The geometrical approach clarifies the relationship between several methods traditionally used to estimate systematic uncertainties in high-energy physics. By describing the profiling as an information-orthogonal horizontal lift and the Schur complement as the induced metric on the parameter-of-interest manifold, physical-source uncertainties arise by mapping observation fluctuations to score covectors, raising them with the inverse total information, and pushing the resulting estimator covariance forward to the parameters of interest. This construction clarifies why nuisance-parameter impacts do not generally coincide with repeated-experiment source variances.}

\keywords{information geometry, profile likelihood, nuisance parameters,
uncertainty decomposition, Fisher information, fiber bundles}

\maketitle

\section{Introduction}
\label{sec:introduction}

Profile-likelihood fits are a standard language for measurements and searches
in high-energy physics.  A single likelihood can combine signal regions,
control regions, calibration data, theoretical variations, and auxiliary
measurements, while nuisance parameters encode the dependence of the model on
quantities that are not themselves the target of inference
\cite{CowanEtAl2011Asymptotic}.  The inverse
curvature of the likelihood then provides the covariance of all fitted
parameters in the regular asymptotic regime.

The decomposition of that covariance into named uncertainty sources is more
subtle than the calculation of the total covariance.  The common practice of
fixing a nuisance parameter, or shifting it by a post-fit standard deviation
and refitting, produces an \textit{impact}.  An impact is useful as a sensitivity
diagnostic, but it does not in general equal the repeated-experiment variance
generated by the corresponding auxiliary measurement.  In particular,
impacts need not add in quadrature to the total variance.  Pinto et al.\
made this distinction explicit and proposed shifts of the primary and
auxiliary observables as a source-consistent decomposition
\cite{Pinto2024Uncertainty}.  Their construction reproduces the familiar
covariance representation in the linear Gaussian limit and remains applicable
through finite refits for non-Gaussian likelihoods.

The purpose of this paper is to supply a geometric picture for these
operations.  The geometric picture makes it explicit that the nuisance
parameter and the global observable that constrains it do not live in the same manifold while, at the same time, elucidating the relation between the two.  A nuisance displacement is a tangent vector in a parameter fiber.  A global-observable displacement is a tangent vector in an observation space.  It first changes the score, hence defines a parameter-space covector; only after the inverse information metric raises that covector does it become a displacement of the fitted parameters. Consequently, the two procedures can be related by a local chain rule without
being the same uncertainty decomposition.

Two bundles organize the discussion.  The \textit{parameter bundle}
\(\pi:\mathcal P\to\mathcal M\) has the parameter-of-interest (POI) manifold
\(\mathcal M\) as base and nuisance manifolds as fibers.  Fisher orthogonality
defines a horizontal distribution on this bundle
\cite{CoxReid1987Orthogonality}.
The tangent to the profile-likelihood section is horizontal, and the
Schur complement is the corresponding quotient information metric
\cite{MurphyVanDerVaart2000Profile}.  The
\textit{observation--parameter bundle}
\(q:\mathcal Y\times\mathcal P\to\mathcal Y\) places a copy of the full
parameter manifold above each possible data set.  The fit is a local section,
and its differential propagates statistical and systematic source
covariances.

This paper shows five different geometric constructions:
\begin{enumerate}
\item Profiling is characterized as the minimum-information horizontal lift and its induced metric is derived.
\item Nuisance shifts are described by a conditional-refit.
\item Shifts of global observables are expressed as the sequence
\[
T_a\mathcal A
\longrightarrow T_{\widehat z}^{*}\mathcal P
\longrightarrow T_{\widehat z}\mathcal P
\longrightarrow T_{\widehat\theta}\mathcal M .
\]
\item The global-observable response is factored through the nuisance
displacement induced by the same shift. 
\item The usual post-fit nuisance eigenmodes are generalized to an arbitrary auxiliary metric, exposing a mode-by-mode scale factor between the two impact conventions.
\end{enumerate}

All identities shown in this paper are local.  Global information geometry is
defined by the expected Fisher metric~\cite{Rao1945Information,AmariNagaoka2000Methods,Amari2016InformationGeometry}. Some comments about finite displacements and curvature are made at the end of the article.

\section{Statistical model and the coupled pair of bundles}
\label{sec:setup}

Let the regular parameter manifold be \(\mathcal P\), with adapted local
coordinates
\begin{equation}
 z=(\theta,\eta)=(\theta,\alpha,\gamma).
 \label{eq:parameter-split}
\end{equation}
The coordinates \(\theta=(\theta^1,\ldots,\theta^p)\) are POIs.
The constrained nuisance coordinates \(\alpha=(\alpha^1,\ldots,\alpha^r)\)
have auxiliary observations (usually associated with external calibrations), whereas \(\gamma\) denotes nuisance parameters without such observations (usually associated with background normalizations).  The complete observation is
\begin{equation}
 y=(m,a)\in\mathcal Y=\mathcal D\times\mathcal A ,
 \label{eq:observation-split}
\end{equation}
where \(m\) denotes primary data and \(a\) denotes global observables or
auxiliary measurements.  Systematic uncertainty models in high-energy physics are usually constructed so that \(\mathcal Y\) may be decomposed into independent physical source manifolds \(\mathcal Y_s\).

For negative log-likelihood \(\ell(z\given y)=-\log L(z\given y)\), the maximum-likelihood estimator \(\widehat z(y)\) is locally defined by
\begin{equation}
 d_z\ell(\widehat z(y)\given y)=0.
 \label{eq:score-zero}
\end{equation}
The expected Fisher information is the covariant tensor
\begin{equation}
 g_{ij}(z)
 =
 \E_z\!\left[
 \partial_i\log p(Y\given z)\,
 \partial_j\log p(Y\given z)
 \right].
 \label{eq:fisher}
\end{equation}
When it is positive definite, \(g\) makes the statistical model a
Riemannian manifold.  The observed information at the fit is
\begin{equation}
 G_{ij}
 =
 \left.
 \frac{\partial^2\ell}{\partial z^i\partial z^j}
 \right|_{\widehat z,y},
 \qquad
 V=G^{-1}.
 \label{eq:observed-information}
\end{equation}
We use \(G\) for local fit identities and \(g\) for intrinsic statements.

We assume throughout the main text that the relevant information blocks are
nonsingular, the optimum is interior and unique on the branch being studied,
and all maps are sufficiently smooth.  Boundaries, discrete nuisances,
multimodality, and singular models are discussed in
Sect.~\ref{sec:beyond-local}.

\subsection{The parameter and the observable bundles}

The first projection is
\begin{equation}
 \pi:\mathcal P\longrightarrow\mathcal M,
 \qquad
 \pi(\theta,\eta)=\theta .
 \label{eq:parameter-bundle}
\end{equation}
Its fiber \(\mathcal F_\theta=\pi^{-1}(\theta)\) contains all nuisance
configurations compatible with a fixed value of the POIs.  Locally this is
the product \(\mathcal M\times\mathcal F\). The vertical subspace at \(z\) is
\begin{equation}
 \Vspace_z\mathcal P
 =
 \ker d\pi_z.
 \label{eq:vertical-space}
\end{equation}
In Sect.~\ref{sec:profiling}, \(G\)-orthogonality will provide an Ehresmann connection and an associated horizontal subspace \(\Hspace_z\mathcal P\).

The second bundle is
\begin{equation}
 q:\mathcal E=\mathcal Y\times\mathcal P\longrightarrow\mathcal Y,
 \qquad q(y,z)=y .
 \label{eq:observation-bundle}
\end{equation}
Above every possible observation \(y\) sits a copy of the complete parameter
manifold.  Under the regularity assumptions, the implicit-function theorem
turns Eq.~\eqref{eq:score-zero} into the local estimator section
\begin{equation}
 \widehat s:\mathcal Y\longrightarrow\mathcal E,
 \qquad
 \widehat s(y)=(y,\widehat z(y)),
 \qquad
 q\circ\widehat s=\id_{\mathcal Y}.
 \label{eq:estimator-section}
\end{equation}
Its POI projection is
\(\widehat\theta=\pi\circ\pr_{\mathcal P}\circ\widehat s\).
Equivalently, the two projections form the bundle tower
\begin{equation}
 \mathcal Y\times\mathcal P
 \xrightarrow{\ \id_{\mathcal Y}\times\pi\ }
 \mathcal Y\times\mathcal M
 \xrightarrow{\ \pr_{\mathcal Y}\ }
 \mathcal Y .
 \label{eq:bundle-tower}
\end{equation}
Figure~\ref{fig:double-bundle} displays the structure.

\begin{figure*}[t]
 \centering
 \includegraphics[width=0.92\textwidth]{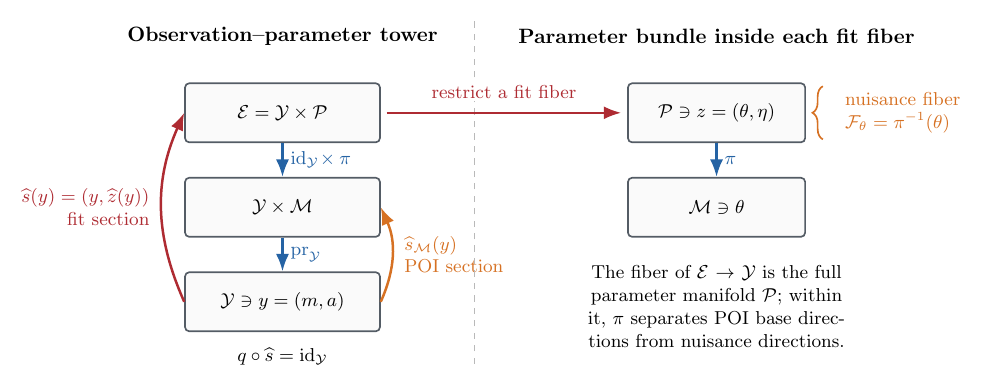}
 \caption{The coupled two-bundle architecture.  The observation--parameter
 bundle places the full parameter manifold above every possible observation.
 Within each such fit fiber, the parameter projection \(\pi\) separates the
 POI base from the nuisance fiber.  The maximum-likelihood fit and its POI
 projection are local sections over observation space.}
 \label{fig:double-bundle}
\end{figure*}

In the case of Gaussian auxiliary likelihood:
\begin{equation}
 \begin{aligned}
  \ell_{\aux}(\alpha\given a)
  &=\frac12(\alpha-a)^{\trans}R(\alpha-a),\\
  \Cov(a)&=R^{-1}.
 \end{aligned}
 \label{eq:gaussian-constraint}
\end{equation}
The nuisance parameter \(\alpha\) is a coordinate in a parameter fiber.
The observed center \(a\) is a coordinate on \(\mathcal A\).  For each fixed
\(a\), the relation \(\alpha=a\) selects an anchor section of the constrained
nuisance space.  Moving \(a\) moves that anchor; it does not directly prescribe
the fitted value of \(\alpha\).

\section{Profiling as horizontal geometry}
\label{sec:profiling}

If we write the local information in POI--nuisance blocks,
\begin{equation}
 G=
 \begin{pmatrix}
  G_{\theta\theta} & G_{\theta\eta}\\
  G_{\eta\theta} & G_{\eta\eta}
 \end{pmatrix}.
 \label{eq:G-blocks}
\end{equation}
The \(G\)-orthogonal complement of
\(\Vspace_z\mathcal P\) defines a horizontal subspace whenever
\(G_{\eta\eta}\) is nonsingular.  A base tangent \(v\in T_\theta\mathcal M\)
has horizontal lift
\begin{equation}
 \Hor_z(v)
 =
 \begin{pmatrix}
  v\\[1mm]
  -G_{\eta\eta}^{-1}G_{\eta\theta}v
 \end{pmatrix}.
 \label{eq:horizontal-lift}
\end{equation}
In adapted coordinates the associated connection one-form may be written
\begin{equation}
 \omega
 =
 \dd\eta+
 G_{\eta\eta}^{-1}G_{\eta\theta}\,\dd\theta ,
 \qquad
 \Hspace_z\mathcal P=\ker\omega_z.
 \label{eq:connection-form}
\end{equation}

\begin{proposition}[Orthogonal profiling and the Schur-complement metric]
\label{prop:profiling}
Let \(\widehat{\widehat\eta}(\theta)\) be the conditional minimizer of
\(\ell(\theta,\eta\given y)\) for fixed \(\theta\).  At a regular point of the
profile section,
\begin{equation}
 D_\theta\widehat{\widehat\eta}
 =
 -G_{\eta\eta}^{-1}G_{\eta\theta}.
 \label{eq:profile-derivative}
\end{equation}
Hence the profile section is tangent to the information-orthogonal horizontal
distribution.  The induced information on POI tangents is
\begin{equation}
 G_{\prof}
 =
 G_{\theta\theta}
 -
 G_{\theta\eta}G_{\eta\eta}^{-1}G_{\eta\theta}.
 \label{eq:profile-schur}
\end{equation}
\end{proposition}

\begin{proof}
Differentiate the vertical stationarity condition
\(d_\eta\ell(\theta,\widehat{\widehat\eta}(\theta)\given y)=0\).
The implicit-function theorem gives
\[
 G_{\eta\theta}\,\dd\theta+
 G_{\eta\eta}\,\dd\widehat{\widehat\eta}=0,
\]
which proves Eq.~\eqref{eq:profile-derivative}.  Substitution of the
horizontal lift into its full squared length gives
\[
 G\bigl(\Hor(v),\Hor(v)\bigr)
 =
 v^{\trans}
 \left(
 G_{\theta\theta}
 -
 G_{\theta\eta}G_{\eta\eta}^{-1}G_{\eta\theta}
 \right)v .
\]
\end{proof}

\begin{figure}[t]
 \centering
 \includegraphics[width=0.6\linewidth]{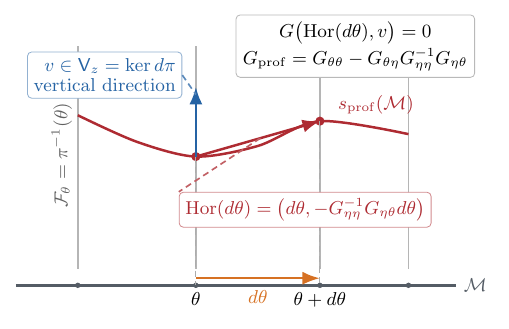}
 \caption{Profiling as an information-orthogonal horizontal lift.  Moving on
 the POI base induces the nuisance displacement that minimizes the local
 information length.  Restricting the full metric to these horizontal
 tangents gives the Schur complement \(G_{\prof}\).}
 \label{fig:profile-connection}
\end{figure}

The same statement has the variational form
\begin{equation}
 v^{\trans}G_{\prof}v
 =
 \min_{w_\eta}
 \begin{pmatrix}v\\w_\eta\end{pmatrix}^{\trans}
 G
 \begin{pmatrix}v\\w_\eta\end{pmatrix}.
 \label{eq:quotient-variational}
\end{equation}
Thus \(G_{\prof}\) is the quotient metric, \textit{i.e.}, the least
information cost of realizing a prescribed POI displacement while allowing
the nuisance fiber to respond.  It is also the efficient information familiar
from profile-likelihood theory~\cite{MurphyVanDerVaart2000Profile}. Block inversion gives
\begin{equation}
 V_{\theta\theta}=G_{\prof}^{-1}.
 \label{eq:profile-covariance}
\end{equation}
This equation is the metric--covariance duality behind the usual profiled
uncertainty.  The metric is covariant and is pulled back along
the horizontal profile section; its inverse is a contravariant covariance.

\section{The estimator section and physical source propagation}
\label{sec:observation-geometry}

The graph of the estimator section is the zero set of the parameter score.
Its tangent is obtained by differentiating Eq.~\eqref{eq:score-zero}.  Define
the mixed score-response map
\begin{equation}
 \mathsf K_y
 \equiv
 -D_y(d_z\ell):
 T_y\mathcal Y\longrightarrow T_{\widehat z}^{*}\mathcal P .
 \label{eq:K-definition}
\end{equation}

\begin{proposition}[Estimator influence on the observation bundle]
\label{prop:influence}
At a regular maximum,
\begin{equation}
 D_y\widehat z
 =
 V^\sharp\mathsf K_y,
 \qquad
 D_y\widehat\theta
 =
 d\pi\,V^\sharp\mathsf K_y ,
 \label{eq:estimator-influence}
\end{equation}
where \(V^\sharp:T_{\widehat z}^{*}\mathcal P\to
T_{\widehat z}\mathcal P\) is the index-raising map represented by
\(V=G^{-1}\).
\end{proposition}

\begin{proof}
Differentiation of the score equation gives
\[
 G\,D_y\widehat z+D_y(d_z\ell)=0.
\]
Multiplication by \(V\) and the definition of \(\mathsf K_y\) yield
Eq.~\eqref{eq:estimator-influence}.
\end{proof}

Equation~\eqref{eq:estimator-influence} gives a precise geometric sequence:
\begin{equation}
 T_y\mathcal Y
 \xrightarrow{\ \mathsf K_y\ }
 T_{\widehat z}^{*}\mathcal P
 \xrightarrow{\ V^\sharp\ }
 T_{\widehat z}\mathcal P
 \xrightarrow{\ d\pi\ }
 T_{\widehat\theta}\mathcal M .
 \label{eq:source-sequence}
\end{equation}
An observation displacement changes the score covector.  The inverse
information converts that covector into a parameter displacement, after which
the parameter projection selects the POI component.

If a source block \(y_s\) has covariance tensor \(\Sigma_s\), its local POI
covariance contribution is the pushforward
\begin{equation}
 \Sigma_\theta^{[s]}
 =
 J_s\Sigma_sJ_s^{\trans},
 \qquad
 J_s=D_{y_s}\widehat\theta .
 \label{eq:source-pushforward}
\end{equation}

The expected information gives an especially compact expression.  Suppose
the likelihood is a product over conditionally independent physical sources
\(s\).  Their score covariances add,
\begin{equation}
 g=\sum_s g_s,
 \qquad
 g_s=\Cov_z\!\left[d_z\log p(Y_s\given z)\right].
 \label{eq:fisher-source-sum}
\end{equation}
The asymptotic covariance carried from source \(s\) to the POI manifold is
\begin{equation}
 \Sigma_{\theta}^{[s]}
 =
 d\pi\,g^{-1}g_sg^{-1}d\pi^{\trans}.
 \label{eq:fisher-source-decomposition}
\end{equation}
Each source supplies a covariant score metric \(g_s\).  The inverse
\emph{total} metric raises both indices and turns it into an estimator
covariance; \(d\pi\) then maps that covariance to the POI base.  Summing over
all sources recovers
\begin{equation}
 \sum_s\Sigma_\theta^{[s]}
 =
 d\pi\,g^{-1}gg^{-1}d\pi^{\trans}
 =
 (g^{-1})_{\theta\theta}.
 \label{eq:fisher-source-closure}
\end{equation}
This identity makes the conceptual difference between parameters and sources
particularly sharp.  Nuisance coordinates label directions in the parameter
manifold; uncertainty components are defined by stochastic score metrics on
the observation manifold.

\subsection{Global-observable shifts}

For the Gaussian auxiliary likelihood in Eq.~\eqref{eq:gaussian-constraint},
a displacement \(da\in T_a\mathcal A\) produces
\begin{equation}
 \mathsf K_a\,da
 =
 \begin{pmatrix}
  0\\ R\,da\\ 0
 \end{pmatrix}
 \in T_{\widehat z}^{*}\mathcal P ,
 \label{eq:GO-score-covector}
\end{equation}
where the blocks correspond to \((\theta,\alpha,\gamma)\).  The injected
covector is supported on constrained nuisance directions.  Equations
\eqref{eq:estimator-influence} and \eqref{eq:GO-score-covector} give
\begin{align}
 d\widehat\theta_{\GO}
 &=V_{\theta\alpha}R\,da,
 \label{eq:GO-theta-response}\\
 d\widehat\alpha_{\GO}
 &=V_{\alpha\alpha}R\,da,
 \label{eq:GO-alpha-response}\\
 d\widehat\gamma_{\GO}
 &=V_{\gamma\alpha}R\,da.
 \label{eq:GO-gamma-response}
\end{align}
Figure~\ref{fig:global-observable} shows both the anchor motion and its
tangent--cotangent interpretation.

\begin{figure*}[t]
 \centering
 \includegraphics[width=0.96\textwidth]{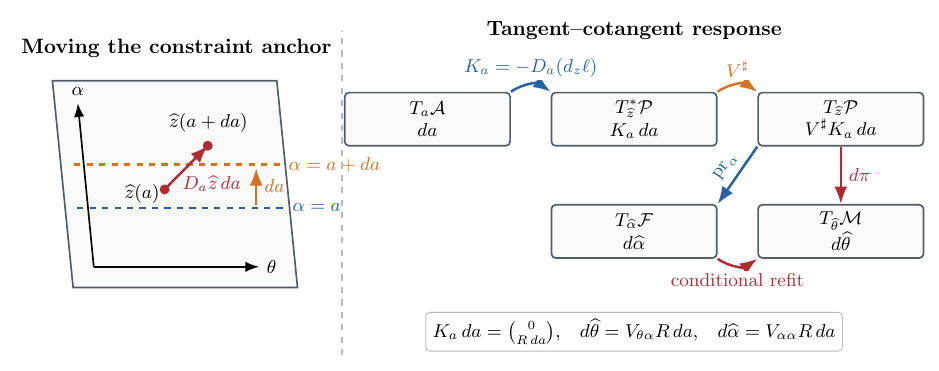}
 \caption{Shifting a global observable translates the anchor of an
 auxiliary constraint. It does not directly shift the fitted nuisance
 coordinate.  The mixed score derivative maps the observation tangent
 into a nuisance-supported covector.  The inverse information raises this
 covector into a full parameter displacement, whose POI and nuisance
 components can then be projected.}
 \label{fig:global-observable}
\end{figure*}

Since \(\Cov(a)=R^{-1}\), the auxiliary-source covariance is
\begin{equation}
 \Sigma_\theta^{[\GO]}
 =
 V_{\theta\alpha}R\,V_{\alpha\theta}.
 \label{eq:GO-covariance}
\end{equation}
For standardized independent constraints, \(R=I\), and a single source \(a_r\)
contributes the rank-one tensor
\begin{equation}
 \Sigma_{\theta}^{[a_r]}
 =
 v_rv_r^{\trans},
 \qquad
 (v_r)_p=V_{\theta^p\alpha^r}.
 \label{eq:GO-rank-one}
\end{equation}
This is the covariance form of the shifted-global-observable construction of
Ref.~\cite{Pinto2024Uncertainty}.  The sign of \(v_r\) depends on how the
nuisance response is parameterized, but it is irrelevant since the outer product does not. For a non-Gaussian auxiliary term, Eq.~\eqref{eq:source-sequence} remains the local formula with the actual mixed derivative \(\mathsf K_a\). 

\section{Nuisance shifts and impacts as conditional geometry}
\label{sec:np-shifts}

Choose a nuisance block \(\alpha\) to be fixed, and collect every parameter that remains free into \(u\), including all POIs and any other nuisance parameters.  Partition the information and its inverse as
\begin{equation}
 G=
 \begin{pmatrix}
  G_{uu}&G_{u\alpha}\\
  G_{\alpha u}&G_{\alpha\alpha}
 \end{pmatrix},
 \qquad
 V=G^{-1}.
 \label{eq:u-alpha-partition}
\end{equation}
For fixed \(\alpha\), let \(u^\star(\alpha)\) denote the conditional optimum.

\begin{proposition}[Conditional-refit response]
\label{prop:conditional-refit}
The differential of the conditional optimum is
\begin{equation}
 D_\alpha u^\star
 =
 -G_{uu}^{-1}G_{u\alpha}
 =
 V_{u\alpha}V_{\alpha\alpha}^{-1}.
 \label{eq:conditional-response}
\end{equation}
The corresponding POI response is
\begin{equation}
 J_{\theta\leftarrow\alpha}
 =
 V_{\theta\alpha}V_{\alpha\alpha}^{-1}.
 \label{eq:NP-response}
\end{equation}
\end{proposition}

\begin{proof}
Differentiate \(d_u\ell(u^\star(\alpha),\alpha)=0\) to obtain the first
equality.  The second follows from the standard block-inverse identity.
\end{proof}
This result was used for quick evaluation of nuisance impacts using auto-differentiation in~\cite{ATLAS:2025clx}. The geometry is worth stating carefully.  A nuisance displacement is vertical
for \(\pi:\mathcal P\to\mathcal M\), and therefore
\(d\pi(0,d\alpha)=0\). The conditional-refit map embeds it as the tangent
\begin{equation}
 d\alpha
 \longmapsto
 \bigl(D_\alpha u^\star\,d\alpha,d\alpha\bigr)
 \in T_z\mathcal P ,
 \label{eq:conditional-embedding}
\end{equation}
whose POI component is then selected.  Equivalently,
Eq.~\eqref{eq:conditional-response} is the horizontal lift for the
role-reversed local fibration in which \(\alpha\) is regarded as the base and
\(u\) as the fiber.

A traditional post-fit nuisance shift assigns
\begin{equation}
 d\alpha=V_{\alpha\alpha}^{1/2}\xi,
 \qquad
 \xi^{\trans}\xi=1,
 \label{eq:one-postfit-sigma}
\end{equation}
and reports
\(J_{\theta\leftarrow\alpha}d\alpha\). A particularly clear choice of basis is discussed in Sec.~\ref{sec:eigenmodes}. This is a sensitivity to a chosen
parameter-space displacement.

\subsection{Profiling penalty by freezing/removing nuisance parameters}

A similar method used in high-energy physics to estimate the impact of a set \(\alpha\) of nuisance parameters is to freeze them and refit the model. The covariance of \(u\) conditional on fixed \(\alpha\) is
\(G_{uu}^{-1}\).  The same block identity used above yields:
\begin{equation}
 V_{uu}
 =
 G_{uu}^{-1}
 +
 V_{u\alpha}V_{\alpha\alpha}^{-1}V_{\alpha u}.
 \label{eq:conditional-marginal-covariance}
\end{equation}
Consequently the POI covariance increase caused by allowing \(\alpha\) to
float is
\begin{equation}
 \Delta V_\theta^{[\prof\,\alpha]}
 =
 V_{\theta\alpha}V_{\alpha\alpha}^{-1}V_{\alpha\theta}
 =
 J_{\theta\leftarrow\alpha}
 V_{\alpha\alpha}
 J_{\theta\leftarrow\alpha}^{\trans}.
 \label{eq:profiling-penalty}
\end{equation}
This is a genuine pushforward of the \textit{post-fit nuisance covariance}
through the conditional response.  It quantifies how much local
POI covariance is lost by profiling this fitted parameter subspace rather
than fixing it.

The related freeze/remove impact compares two inverse horizontal metrics:
the full profile geometry and the geometry of a restricted model in which a
nuisance direction is absent or fixed.  It therefore compares two
experiments or two submodels.  It is not an additive decomposition of the
source covariance of one complete repeated experiment~\cite{Pinto2024Uncertainty}.

There is nevertheless an exact local relation to the one-post-fit-\(\sigma\)
shifts.  For any complete orthonormal basis \(\{e_i\}\) of the whitened
nuisance tangent space, define
\begin{equation}
 \begin{aligned}
  \Delta\theta_i^{1\sigma}
  &=J_{\theta\leftarrow\alpha}
    V_{\alpha\alpha}^{1/2}e_i,\\
  \sum_i\Delta\theta_i^{1\sigma}
  (\Delta\theta_i^{1\sigma})^{\trans}
  &=J_{\theta\leftarrow\alpha}V_{\alpha\alpha}
    J_{\theta\leftarrow\alpha}^{\trans}
   =\Delta V_\theta^{[\prof\,\alpha]}.
 \end{aligned}
 \label{eq:profiling-mode-impacts}
\end{equation}
Thus, for one POI, the freeze/float variance penalty is the quadrature sum of
a complete set of such mode impacts; for a single scalar nuisance it is the
square of its one-post-fit-\(\sigma\) impact.  This is a local quadratic
identity for a complete whitened basis.  One-at-a-time shifts of correlated
named nuisance coordinates need not form that basis, and the individual
terms are not automatically physical source components.

\section{The local relationship between nuisance and global-observable shifts}
\label{sec:local-relation}

The two response maps meet at the nuisance displacement that a
global-observable fluctuation actually induces.

\begin{proposition}[Factorization of global-observable and nuisance shifts]
\label{prop:factorization}
For the nuisance block constrained by Eq.~\eqref{eq:gaussian-constraint},
\begin{equation}
 J_{\theta\leftarrow a}
 =
 J_{\theta\leftarrow\alpha}
 J_{\alpha\leftarrow a}
 =
 \left(V_{\theta\alpha}V_{\alpha\alpha}^{-1}\right)
 \left(V_{\alpha\alpha}R\right)
 =
 V_{\theta\alpha}R .
 \label{eq:factorization}
\end{equation}
Thus the NP-shift and shifted-global-observable procedures give the same
first-order POI displacement if the NP shift is chosen to be
\begin{equation}
 d\alpha=d\widehat\alpha_{\GO}
 =V_{\alpha\alpha}R\,da .
 \label{eq:matching-NP-shift}
\end{equation}
\end{proposition}

\begin{proof}
Insert Eq.~\eqref{eq:GO-alpha-response} into
Eq.~\eqref{eq:NP-response}.  The factors of \(V_{\alpha\alpha}\) cancel.
\end{proof}

\begin{figure}[t]
 \centering
 \includegraphics[width=0.7\linewidth]{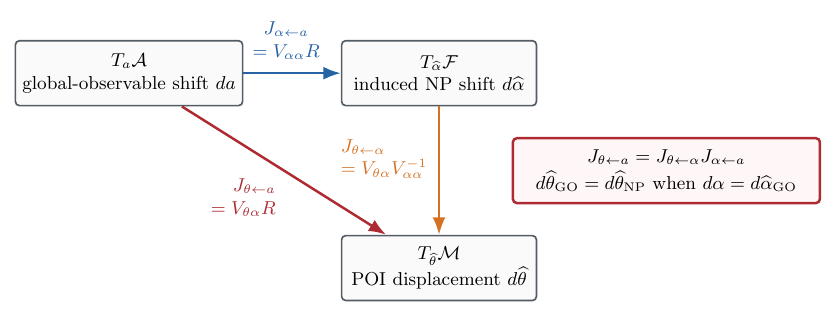}
 \caption{The commuting triangle relating global-observable and nuisance
 shifts.  The direct source-to-POI response equals the conditional nuisance
 response evaluated on the nuisance displacement generated by that same
 source.}
 \label{fig:local-relation}
\end{figure}

In the scalar standardized case,
\begin{equation}
 R=1,\qquad V_{\alpha\alpha}=\sigma_\alpha^2 ,
 \label{eq:scalar-standardized}
\end{equation}
a unit global-observable shift gives
\begin{equation}
 d\widehat\alpha_{\GO}=\sigma_\alpha^2,
 \qquad
 d\widehat\theta_{\GO}=V_{\theta\alpha}.
 \label{eq:scalar-GO}
\end{equation}
A one-post-fit-standard-deviation nuisance shift instead gives
\begin{equation}
 d\widehat\theta_{\NP}^{1\sigma}
 =
 \frac{V_{\theta\alpha}}{\sigma_\alpha}.
 \label{eq:scalar-NP}
\end{equation}
Therefore
\begin{equation}
 d\widehat\theta_{\GO}^{1\sigma(a)}
 =
 \sigma_\alpha\,
 d\widehat\theta_{\NP}^{1\sigma(\alpha)}.
 \label{eq:scalar-scale-relation}
\end{equation}
When the primary data constrain the nuisance strongly,
\(\sigma_\alpha\ll1\), a one-post-fit-\(\sigma\) nuisance shift is much larger
than the nuisance motion generated by a one-\(\sigma\) auxiliary fluctuation.

The distinction is equally transparent at covariance level.  If nuisance
displacements are assigned their full post-fit covariance, then
\begin{equation}
 \Sigma_\theta^{[\NP]}
 =
 V_{\theta\alpha}V_{\alpha\alpha}^{-1}V_{\alpha\theta},
 \label{eq:NP-propagated-covariance}
\end{equation}
which is the profiling penalty in Eq.~\eqref{eq:profiling-penalty}.  By
contrast, global-observable fluctuations induce only
\begin{equation}
 \Cov(d\widehat\alpha_{\GO})
 =
 V_{\alpha\alpha}R\,V_{\alpha\alpha}.
 \label{eq:GO-induced-alpha-covariance}
\end{equation}
Pushing this covariance through the conditional map gives
\begin{align}
 &J_{\theta\leftarrow\alpha}
 \Cov(d\widehat\alpha_{\GO})
 J_{\theta\leftarrow\alpha}^{\trans}
 \nonumber\\
 &\hspace{15mm}
 =
 V_{\theta\alpha}R\,V_{\alpha\theta}
 =
 \Sigma_\theta^{[\GO]}.
 \label{eq:covariance-factorization}
\end{align}
The two covariance constructions become identical only after the nuisance
shift covariance is chosen to be the part generated by auxiliary-source
fluctuations.

A global-observable fluctuation first injects the nuisance
score covector \(\kappa_\alpha=R\,da\).  Since
\(\Cov(da)=R^{-1}\), its covariance is \(\Cov(\kappa_\alpha)=R\).
The inverse information then maps this score kick to the fitted nuisance
displacement
\(d\widehat\alpha_{\GO}=V_{\alpha\alpha}\kappa_\alpha\). Its pushforward is:

\begin{equation}
 \Cov(d\widehat\alpha_{\GO})
 =V_{\alpha\alpha}\Cov(\kappa_\alpha)
  V_{\alpha\alpha}^{\trans}
 =V_{\alpha\alpha}R\,V_{\alpha\alpha}.
 \label{eq:GO-score-kick-pushforward}
\end{equation}
This tensor isolates the nuisance motion generated by the auxiliary source.
The full post-fit covariance \(V_{\alpha\alpha}\), by contrast, contains
nuisance motion induced by all primary and auxiliary sources.  In the scalar
standardized case the susceptibility is
\(d\widehat\alpha/da=\sigma_\alpha^2\), so a unit-variance anchor fluctuation
produces variance \(\sigma_\alpha^4\) in the fitted nuisance coordinate.

Note that fixing all constrained nuisance coordinates defines a conditional
covariance, not in general the repeated-experiment covariance generated by the
primary data.  From Eq.~\eqref{eq:conditional-marginal-covariance}, its POI
block is
\begin{equation}
 V_{\theta\theta}^{[\mathrm{fix}\,\alpha]}
 =V_{\theta\theta}
 -V_{\theta\alpha}V_{\alpha\alpha}^{-1}V_{\alpha\theta}.
 \label{eq:fixed-versus-total-covariance}
\end{equation}
By contrast, when the source closure holds and the constrained coordinates
are anchored by the Gaussian auxiliary block, the primary-data component is
\begin{equation}
 \Sigma_\theta^{[\stat]}
 =V_{\theta\theta}-\Sigma_\theta^{[\GO]}
 =V_{\theta\theta}-V_{\theta\alpha}R V_{\alpha\theta}.
 \label{eq:source-statistical-covariance}
\end{equation}
The two subtractions involve different nuisance covariances.  Freezing
\(\alpha\) removes not only its response to the auxiliary observations, but
also nuisance motion induced by fluctuations of the primary data.  Therefore
the fixed-nuisance result may be reported as an operational ``statistical-only'' uncertainty, but it should not be identified with the source-consistent statistical component in Eq.~\eqref{eq:source-statistical-covariance} except in special cases where the two expressions coincide.

\subsection{An exact endpoint identity for a finite shift}

The local matrix relation has a useful nonlinear parent.  Suppose the
likelihood factorizes as
\begin{equation}
 \ell(u,\alpha\given m,a)
 =
 \ell_0(u,\alpha\given m)
 +
 \ell_{\aux}(\alpha\given a),
 \label{eq:separable-likelihood}
\end{equation}
so that \(a\) does not enter the \(u\)-score.  Define
\(a'=a+\Delta a\), and first perform the ordinary shifted-global-observable
fit with both parameter blocks floating:
\begin{equation}
 (\widehat u_{a'},\widehat\alpha_{a'})
 =\argmin_{u,\alpha}
 \left\{\ell_0(u,\alpha\given m)
 +\ell_{\aux}(\alpha\given a')\right\}.
 \label{eq:finite-all-floating-fit}
\end{equation}
Only after this first fit do we freeze the nuisance block, at the endpoint
\(\widehat\alpha_{a'}\) that the shifted-observable fit itself selected, and
conditionally refit \(u\) to the same primary data:
\begin{equation}
 u^\star(\widehat\alpha_{a'};m)
 =\argmin_u\ell_0(u,\widehat\alpha_{a'}\given m).
 \label{eq:finite-conditional-map}
\end{equation}
The auxiliary term is constant with respect to \(u\), so the \(u\)-stationarity
condition in the all-floating fit is exactly the one imposed by this second,
conditional fit.  On any unique regular solution branch,
\begin{equation}
 \widehat u_{a'}=u^\star(\widehat\alpha_{a'};m).
 \label{eq:finite-composition}
\end{equation}
Thus a finite shifted-global-observable fit can be reproduced by fixing the
nuisance block to the \textit{value induced by that fit} and conditionally
refitting \(u\).  What is only local is the replacement of this nonlinear
composition by the constant matrices in Eq.~\eqref{eq:factorization}, and the
simple scale comparison between conventional one-\(\sigma\) shifts.

\section{Multiple constrained nuisances and invariant eigenmodes}
\label{sec:eigenmodes}

An ordinary diagonalization of \(V_{\alpha\alpha}\) depends on the units and
coordinate normalization chosen for the nuisance parameters.  The auxiliary
precision \(R\) supplies the natural covariant metric with which to define
dimensionless modes.  Whiten the auxiliary geometry and diagonalize
\begin{equation}
 \begin{aligned}
  W_\alpha
  &\equiv R^{1/2}V_{\alpha\alpha}R^{1/2}\\
  &=Q\Lambda Q^{\trans},\\[-1mm]
  \Lambda&=\diag(\lambda_1,\ldots,\lambda_r).
 \end{aligned}
 \label{eq:whitened-diagonalization}
\end{equation}
Define the principal coordinates
\begin{equation}
 \beta=Q^{\trans}R^{1/2}\alpha,
 \qquad
 b=Q^{\trans}R^{1/2}a.
 \label{eq:principal-coordinates}
\end{equation}
Then
\begin{equation}
 V_{\beta\beta}=\Lambda,
 \qquad
 \Cov(b)=I .
 \label{eq:principal-covariances}
\end{equation}
Equivalently, the physical nuisance directions
\begin{equation}
 p_i=R^{-1/2}q_i
 \label{eq:physical-mode}
\end{equation}
solve the generalized eigenproblem
\begin{equation}
 V_{\alpha\alpha}R\,p_i=\lambda_i p_i,
 \qquad
 p_i^{\trans}Rp_j=\delta_{ij}.
 \label{eq:generalized-eigenproblem}
\end{equation}

Let
\begin{equation}
 C=V_{\theta\alpha}R^{1/2}Q,
 \qquad c_i=C e_i .
 \label{eq:C-mode-coupling}
\end{equation}
The global-observable response and induced nuisance response become
\begin{equation}
 \frac{\partial\widehat\theta}{\partial b}=C,
 \qquad
 \frac{\partial\widehat\beta}{\partial b}=\Lambda.
 \label{eq:diagonal-GO-response}
\end{equation}
A unit auxiliary fluctuation in mode \(i\) therefore gives
\begin{equation}
 \Delta\theta_{\GO,i}=c_i,
 \qquad
 \Delta\beta_{\GO,i}=\lambda_i e_i .
 \label{eq:mode-GO}
\end{equation}
A one-post-fit-standard-deviation shift in the corresponding nuisance mode is
\(\Delta\beta_i=\sqrt{\lambda_i}e_i\), and hence
\begin{equation}
 \Delta\theta_{\NP,i}^{1\sigma}
 =
 \frac{c_i}{\sqrt{\lambda_i}}.
 \label{eq:mode-NP}
\end{equation}
The mode-by-mode relation is
\begin{equation}
 \Delta\theta_{\GO,i}
 =
 \sqrt{\lambda_i}\,
 \Delta\theta_{\NP,i}^{1\sigma},
 \qquad
 \Sigma_{\theta,i}^{[\GO]}
 =
 \lambda_i\Sigma_{\theta,i}^{[\NP]}.
 \label{eq:mode-scale}
\end{equation}
Figure~\ref{fig:eigenmodes} depicts the two lengths in whitened nuisance
space and their images in the POI tangent space.

\begin{figure}[t]
 \centering
 \includegraphics[width=0.7\linewidth]{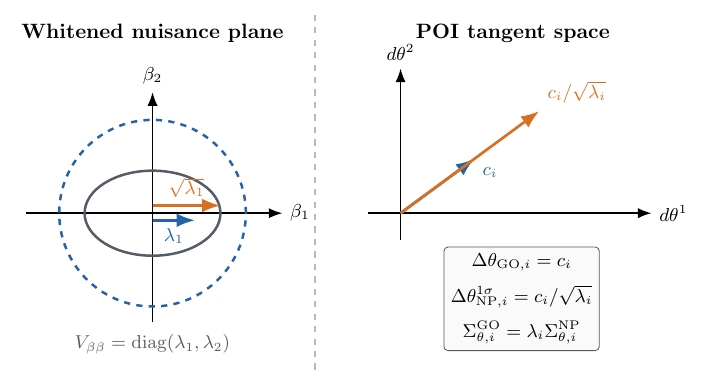}
 \caption{Mode-wise comparison after whitening by the auxiliary metric.
 The post-fit nuisance covariance has principal squared radii \(\lambda_i\).
 A unit auxiliary fluctuation moves the fitted mode by \(\lambda_i\), whereas
 a one-post-fit-\(\sigma\) nuisance shift has length \(\sqrt{\lambda_i}\).
 Their POI images differ by the same factor.}
 \label{fig:eigenmodes}
\end{figure}

Summing the rank-one mode tensors gives
\begin{align}
 \Sigma_\theta^{[\NP]}
 &=
 C\Lambda^{-1}C^{\trans}
 =
 \sum_i\frac{c_ic_i^{\trans}}{\lambda_i},
 \label{eq:NP-mode-sum}\\
 \Sigma_\theta^{[\GO]}
 &=
 CC^{\trans}
 =
 \sum_i c_ic_i^{\trans}.
 \label{eq:GO-mode-sum}
\end{align}
If the total expected information has the form
\begin{equation}
 g=g_0+E_\alpha R E_\alpha^{\trans},
 \qquad g_0\succeq0,
 \label{eq:lambda-bound-assumptions}
\end{equation}
then \(0<\lambda_i\leq1\).  The eigenvalue is the remaining post-fit variance
of mode \(i\), measured in units of its auxiliary variance.

Equation~\eqref{eq:NP-mode-sum} is equivalently
\begin{equation}
 \Sigma_\theta^{[\NP]}
 =\sum_i
 \Delta\theta_{\NP,i}^{1\sigma}
 (\Delta\theta_{\NP,i}^{1\sigma})^{\trans}.
 \label{eq:NP-impact-quadrature}
\end{equation}
For one POI, this is precisely a sum in quadrature of the principal-mode NP
impacts.  It is a valid additive decomposition of the profiling geometry,
but not, in general, a decomposition by physical uncertainty source.  The
rotated modes are fit-dependent combinations of the named auxiliary
measurements.  Nondegenerate generalized eigendirections are fixed by the pair
\((V_{\alpha\alpha},R)\), but within a degenerate eigenspace their individual
rank-one terms can be redistributed by a rotation.  In either case, a mode
label is not a source label.  Source-by-source reporting should remain in the
physically defined observation basis unless a principal-mode summary is
explicitly intended.

\section{The linear Gaussian model and the covariance representation}
\label{sec:gaussian-model}

The linear Gaussian model makes contact with the formulas of
Ref.~\cite{Pinto2024Uncertainty} and provides a direct check of the geometric
construction.  Let
\begin{equation}
 \ell(\theta,\alpha\given m,a)
 =
 \frac12 r^{\trans}S^{-1}r
 +
 \frac12(\alpha-a)^{\trans}R(\alpha-a),
 \label{eq:linear-gaussian-likelihood}
\end{equation}
with
\begin{equation}
 r=m-t_0-X\theta-\Gamma\alpha .
 \label{eq:linear-residual}
\end{equation}
Here \(S\) is the statistical covariance of the primary measurements,
\(X\) is the POI design matrix, and \(\Gamma\) carries nuisance shifts into
measurement space.

Writing \(\beta=\alpha-a\), define the shifted residual
\begin{equation}
 r_a(\theta)=m-t_0-X\theta-\Gamma a.
 \label{eq:shifted-linear-residual}
\end{equation}
Profiling \(\beta\) gives
\begin{equation}
 \ell_{\prof}(\theta)
 =
 \frac12 r_a(\theta)^{\trans}C^{-1}r_a(\theta)
 +\mathrm{const.},
 \label{eq:covariance-representation}
\end{equation}
where
\begin{equation}
 C=S+\Gamma R^{-1}\Gamma^{\trans}.
 \label{eq:total-measurement-covariance}
\end{equation}
The estimator is
\begin{equation}
 \begin{aligned}
  \widehat\theta&=L(m-t_0-\Gamma a),\\
  L&=(X^{\trans}C^{-1}X)^{-1}X^{\trans}C^{-1}.
 \end{aligned}
 \label{eq:linear-estimator}
\end{equation}
Consequently,
\begin{equation}
 J_{\theta\leftarrow m}=L,
 \qquad
 J_{\theta\leftarrow a}=-L\Gamma .
 \label{eq:linear-source-Jacobians}
\end{equation}
The source covariance pushforwards are
\begin{align}
 \Sigma_\theta^{[\stat]}
 &=LSL^{\trans},
 \label{eq:linear-stat-covariance}\\
 \Sigma_\theta^{[\syst]}
 &=L\Gamma R^{-1}\Gamma^{\trans}L^{\trans},
 \label{eq:linear-syst-covariance}\\
 \Sigma_\theta^{[\mathrm{tot}]}
 &=LCL^{\trans}
 =(X^{\trans}C^{-1}X)^{-1}.
 \label{eq:linear-total-covariance}
\end{align}
They close exactly because \(C\) is the sum of the source covariances.

Comparing Eq.~\eqref{eq:linear-source-Jacobians} with the general
global-observable response gives
\begin{equation}
 V_{\theta\alpha}R=-L\Gamma.
 \label{eq:cross-covariance-identity}
\end{equation}
For standardized independent nuisances, \(R=I\), the POI--nuisance
cross-covariance is therefore precisely the signed systematic response vector
to a unit global-observable shift.  Its outer product reproduces the
systematic covariance component.  This is the geometric content of the
cross-covariance and shifted-observable identities derived in
Ref.~\cite{Pinto2024Uncertainty}.  It also shows why the answer is not obtained
by propagating the full post-fit covariance \(V_{\alpha\alpha}\): the latter
includes nuisance motion induced by both primary and auxiliary data. 
This is the linear-Gaussian realization of the general score-kick argument
following Eq.~\eqref{eq:GO-score-kick-pushforward}.

\section{Unconstrained nuisance parameters}
\label{sec:unconstrained}

An unconstrained nuisance coordinate \(\gamma\) remains a vertical direction
of the parameter bundle.  It may be well measured by primary data and can be profiled in the usual way.  What it lacks is a corresponding
auxiliary-observation direction and auxiliary score metric.  There is no
\(a_\gamma\), no precision \(R_\gamma\), and hence no
global-observable covariance of the form
\(V_{\theta\gamma}R_\gamma V_{\gamma\theta}\).

If \(\gamma\) is fixed and all other coordinates \(u\) are refitted, its local
profiling penalty is
\begin{equation}
 \Delta V_\theta^{[\prof\,\gamma]}
 =
 V_{\theta\gamma}V_{\gamma\gamma}^{-1}V_{\gamma\theta}.
 \label{eq:unconstrained-profiling-penalty}
\end{equation}
This is meaningful as a comparison between floating and fixing a parameter
subspace.  It is not automatically a contribution from a physical uncertainty
source.  

When \(\gamma\) is inferred from control-region counts or other
primary observations, the source-consistent statistical component is instead
obtained by pushing the covariance of those observations through
Eq.~\eqref{eq:source-pushforward}. For example, consider a signal-region count and a control-region count,
\begin{equation}
 n_{\mathrm{SR}}\sim
 \operatorname{Pois}(\theta s+\gamma b_{\mathrm{SR}}),
 \qquad
 n_{\mathrm{CR}}\sim
 \operatorname{Pois}(\gamma b_{\mathrm{CR}}),
 \label{eq:unconstrained-control-region-example}
\end{equation}
with no auxiliary constraint on the background normalization \(\gamma\).
Both counts are primary observations.  Fixing \(\gamma\) measures the local
profiling penalty
\(V_{\theta\gamma}^{2}/V_{\gamma\gamma}\), whereas the statistical variance
carried by the control-region count is
\begin{equation}
 \Sigma_\theta^{[\mathrm{CR}]}
 =
 \left(\frac{\partial\widehat\theta}{\partial n_{\mathrm{CR}}}\right)^2
 \Var(n_{\mathrm{CR}}).
 \label{eq:control-region-source-covariance}
\end{equation}
These quantities need not be equal: the first asks what is gained by fixing a
fitted parameter, while the second fluctuates a specified observation source
and allows the complete fit, including \(\gamma\), to respond. In practice, it is not always possible to identify an individual control region that constrains each floating normalization.

For several unconstrained nuisances, the contribution assigned to an
individual coordinate depends on the chosen basis and, for sequential
freeze/remove procedures, potentially on the ordering.  A specified nuisance
\emph{subspace} has an invariant profiling penalty under reparameterizations
within that subspace, but an NP-by-NP attribution is not intrinsic.
Constrained and unconstrained directions should not be mixed in the
generalized eigenproblem of Sect.~\ref{sec:eigenmodes} when the goal is
auxiliary-source attribution.

\section{Beyond the local Gaussian regime}
\label{sec:beyond-local}

Along a finite global-observable path \(a(t)\), the local response evolves as
\begin{equation}
 \frac{\dd\widehat z}{\dd t}
 =
 V\bigl(\widehat z(t),a(t)\bigr)^\sharp
 \mathsf K_a\bigl(\widehat z(t),a(t)\bigr)
 \frac{\dd a}{\dd t}.
 \label{eq:finite-response-ode}
\end{equation}
The eigenbasis \(Q\), eigenvalues \(\Lambda\), and POI coupling \(C\) all vary
along the path.  A fixed best-fit eigenmode therefore mixes with other modes
under a finite displacement even in a flat manifold but curvature adds
path-dependence to parallel identification of frames.  

At boundaries, at singular points, or when the optimizer jumps between
modes, a smooth estimator section may not exist.  In those regimes, the
bundle picture remains a useful stratified intuition, but the differential
identities cannot be used without a problem-specific treatment. In practice that means that the non-regular analogue of the local pushforward is ordinarily a finite-refit or ensemble construction, rather than a modified matrix identity.

\section{Conclusions}
\label{sec:dictionary}

Table~\ref{tab:dictionary} summarizes the operations that are often grouped
under the word impact.

\begin{table*}[t]
\caption{Geometric classification of common uncertainty operations.}
\label{tab:dictionary}
\centering
\begin{tabular}{>{\raggedright\arraybackslash}p{0.16\textwidth}
                >{\raggedright\arraybackslash}p{0.18\textwidth}
                >{\raggedright\arraybackslash}p{0.25\textwidth}
                >{\raggedright\arraybackslash}p{0.30\textwidth}}
\toprule
Procedure & Initial object & Geometric operation & Question answered\\
\midrule
Profiling
& \(d\theta\in T\mathcal M\)
& Information-orthogonal horizontal lift; restrict the metric to the
  profile section
& What information remains after the nuisance fiber is allowed to relax?\\
\addlinespace
NP shift and refit
& Assigned \(d\alpha\in T\mathcal F\)
& Conditional-refit embedding followed by POI projection,
  \(J_{\theta\leftarrow\alpha}d\alpha\)
& How sensitive is the fit to this chosen fitted-parameter displacement?\\
\addlinespace
Freeze/remove impact
& Two parameter submodels or sections
& Compare inverse horizontal metrics with a nuisance block floating and fixed
& What is the local covariance cost of allowing this subspace to float?\\
\addlinespace
Global-observable shift
& \(da\in T\mathcal A\)
& Inject a score covector, raise it with \(V\), and project to POIs
& What repeated-experiment variance is generated by this auxiliary source?\\
\addlinespace
Primary-data shift
& \(dm\in T\mathcal D\)
& Push the data covariance through \(D_m\widehat\theta\)
& What repeated-experiment statistical variance is generated by these data?\\
\bottomrule
\end{tabular}
\end{table*}

The following steps are the recommended minimum for a source-consistent
measurement decomposition:
\begin{enumerate}
\item Fit the complete model and compute \(G\) and \(V=G^{-1}\) at the nominal
      optimum.  Within the local Hessian approximation,
      \(V_{\theta\theta}\) is the total POI covariance of the full fit and
      therefore the closure target for the decomposition.  Record whether
      observed or expected information is used;
\item Partition constrained and unconstrained nuisance directions.  Determine
      the auxiliary covariance or precision for each physical source block.
\item Obtain and report source components from finite observation shifts or from
      \(J_s\Sigma_sJ_s^{\trans}\).
\item Report the final statistical-only component as the sum over primary-data
      source blocks,
      \begin{equation}
       \Sigma_\theta^{[\stat]}
       =\sum_{s\in\mathrm{primary}}
       J_s\Sigma_sJ_s^{\trans}
       \approx V_{\theta\theta}
        -\sum_{s\in\mathrm{auxiliary}}\Sigma_\theta^{[s]},
       \label{eq:workflow-statistical-component}
      \end{equation}
      The last comparison is exact for the corresponding expected-information identity and is an appropriate approximation for an observed local implementation. 
\item For an unconstrained nuisance \(\gamma\), do not assign an auxiliary
      source component that does not exist.  Its constraining primary
      observations are already included in
      Eq.~\eqref{eq:workflow-statistical-component}.  If a summary of its
      parameter-space effect is useful, report the profiling penalty in
      Eq.~\eqref{eq:unconstrained-profiling-penalty} and label it as such.
\item Test nonlinearity using multiple shift sizes, both signs, or ensemble
      propagation.  Near boundaries, singularities, or optimizer branch
      changes, replace the local Hessian propagation by the finite-refit or
      simulation-calibrated procedures described in
      Sect.~\ref{sec:beyond-local}.
\end{enumerate}

Whitened nuisance eigenmodes, conventional one-post-fit-\(\sigma\) NP shifts,
and freeze/remove impacts are optional diagnostics; they are not required to
report a physical-source decomposition.  In principal mode \(i\), a unit
auxiliary fluctuation induces a nuisance displacement \(\lambda_i e_i\),
whereas the conventional post-fit NP shift is \(\sqrt{\lambda_i}e_i\).
Therefore a measurement that aims only at source attribution should use the
observation shifts in the required workflow above.  Eigenmodes or NP impacts
may still be reported when sensitivity or profiling diagnostics are useful,
provided they are clearly labeled and kept separate from source components.

The geometric distinction between parameters and observations resolves the
apparent tension among common profile-likelihood uncertainty procedures.
Profiling belongs to the parameter bundle.  The Fisher or observed information
selects horizontal directions orthogonal to the nuisance fibers, and the
Schur complement is the resulting relaxed metric on POI displacements.

Nuisance shift-and-refit impacts belong to a conditional geometry on the same
parameter space.  They propagate a chosen tangent vector through a
conditional optimizer and are naturally suited to sensitivity and profiling
questions. Uncertainty components belong instead to the observation bundle.  A physical source fluctuation changes the parameter score; the inverse total information raises that score covector into an estimator displacement.  Its covariance is then pushed forward to the POI manifold.  In the expected-information
language, each source contributes \(d\pi\,g^{-1}g_sg^{-1}d\pi^{\trans}\),
which sums to the total POI covariance under the information identity.
The shifted-global-observable method is the auxiliary-measurement realization
of this source geometry.

A global-observable shift moves the fitted nuisance parameter, and the resulting POI response factors exactly at first order through the conditional nuisance response.  They agree when the nuisance displacement used in the shift-and-refit procedure is the one actually induced by the auxiliary fluctuation.  A conventional one-post-fit-standard-deviation shift uses a different displacement, explaining the scale mismatch and the different covariance decompositions.

\section*{Data availability}
No datasets were generated or analyzed as part of this study.

\section*{Acknowledgments}
The work of RCLSA is partially supported by the US Department of Energy award DE-SC0010004.

\begingroup
\bibliography{references}

@article{Pinto2024Uncertainty,
  author        = {Pinto, Andr{\'e}s and Wu, Zhibo and Balli, Fabrice and Berger, Nicolas and Boonekamp, Maarten and Chapon, {\'E}milien and Kawamoto, Tatsuo and Malaescu, Bogdan},
  title         = {Uncertainty components in profile likelihood fits},
  journal       = {Eur. Phys. J. C},
  volume        = {84},
  number        = {6},
  pages         = {593},
  year          = {2024},
  doi           = {10.1140/epjc/s10052-024-12877-5},
  eprint        = {2307.04007},
  archivePrefix = {arXiv},
  primaryClass  = {physics.data-an}
}

@article{Rao1945Information,
  author  = {Rao, C. Radhakrishna},
  title   = {Information and the Accuracy Attainable in the Estimation of Statistical Parameters},
  journal = {Bull. Calcutta Math. Soc.},
  volume  = {37},
  pages   = {81--91},
  year    = {1945}
}

@book{Amari2016InformationGeometry,
  author    = {Amari, Shun-ichi},
  title     = {Information Geometry and Its Applications},
  series    = {Applied Mathematical Sciences},
  volume    = {194},
  publisher = {Springer},
  address   = {Tokyo},
  year      = {2016},
  doi       = {10.1007/978-4-431-55978-8},
  isbn      = {978-4-431-55978-8}
}

@book{AmariNagaoka2000Methods,
  author    = {Amari, Shun-ichi and Nagaoka, Hiroshi},
  title     = {Methods of Information Geometry},
  series    = {Translations of Mathematical Monographs},
  volume    = {191},
  publisher = {American Mathematical Society},
  address   = {Providence, RI},
  year      = {2000},
  doi       = {10.1090/mmono/191}
}

@article{CoxReid1987Orthogonality,
  author  = {Cox, D. R. and Reid, N.},
  title   = {Parameter Orthogonality and Approximate Conditional Inference},
  journal = {J. R. Stat. Soc. B},
  volume  = {49},
  number  = {1},
  pages   = {1--18},
  year    = {1987},
  doi     = {10.1111/j.2517-6161.1987.tb01422.x}
}

@article{MurphyVanDerVaart2000Profile,
  author  = {Murphy, S. A. and van der Vaart, A. W.},
  title   = {On Profile Likelihood},
  journal = {J. Am. Stat. Assoc.},
  volume  = {95},
  number  = {450},
  pages   = {449--465},
  year    = {2000},
  doi     = {10.1080/01621459.2000.10474219}
}

@article{CowanEtAl2011Asymptotic,
  author        = {Cowan, Glen and Cranmer, Kyle and Gross, Eilam and Vitells, Ofer},
  title         = {Asymptotic Formulae for Likelihood-Based Tests of New Physics},
  journal       = {Eur. Phys. J. C},
  volume        = {71},
  number        = {2},
  pages         = {1554},
  year          = {2011},
  doi           = {10.1140/epjc/s10052-011-1554-0},
  eprint        = {1007.1727},
  archivePrefix = {arXiv},
  primaryClass  = {physics.data-an}
}

@article{ATLAS:2025clx,
    author = "{ATLAS Collaboration}",
    collaboration = "ATLAS",
    title = "{An implementation of neural simulation-based inference for parameter estimation in ATLAS}",
    eprint = "2412.01600",
    archivePrefix = "arXiv",
    primaryClass = "physics.data-an",
    reportNumber = "CERN-EP-2024-305",
    doi = "10.1088/1361-6633/add370",
    journal = "Rept. Prog. Phys.",
    volume = "88",
    number = "6",
    pages = "067801",
    year = "2025"
}
\endgroup
\appendix

\section{General auxiliary constraints}
\label{app:general-constraints}

In this note, all examples are given with Gaussian constraints, which represent the most common case in high-energy physics. In general, let the auxiliary negative log-likelihood be a smooth function \(c(\alpha,a)\), not necessarily Gaussian.  Define
\begin{equation}
 K_{\alpha a}
 =
 -\left.
 \frac{\partial^2c}{\partial\alpha\,\partial a}
 \right|_{\widehat\alpha,a}.
 \label{eq:general-K}
\end{equation}
Then the local responses are
\begin{equation}
 D_a\widehat\theta
 =
 V_{\theta\alpha}K_{\alpha a},
 \qquad
 D_a\widehat\alpha
 =
 V_{\alpha\alpha}K_{\alpha a}.
 \label{eq:general-aux-response}
\end{equation}
The factorization through the conditional nuisance response remains
\begin{equation}
 D_a\widehat\theta
 =
 \left(V_{\theta\alpha}V_{\alpha\alpha}^{-1}\right)
 D_a\widehat\alpha .
 \label{eq:general-factorization}
\end{equation}
For auxiliary covariance \(\Sigma_a\), the source contribution is
\begin{equation}
 \Sigma_\theta^{[a]}
 =
 V_{\theta\alpha}K_{\alpha a}\Sigma_a
 K_{\alpha a}^{\trans}V_{\alpha\theta}.
 \label{eq:general-aux-covariance}
\end{equation}
The Gaussian location constraint is the special case
\(K_{\alpha a}=R\) and \(\Sigma_a=R^{-1}\).

\end{document}